\documentclass{article}
\usepackage[utf8]{inputenc}

\usepackage[utf8]{inputenc}

\usepackage{comment}

\usepackage{graphicx} 
\usepackage{amsmath,amsthm,verbatim,amssymb,amsfonts,amscd, graphicx}
\usepackage[utf8]{inputenc}
\usepackage{graphics}
\usepackage{hyperref}
\usepackage[baseline]{euflag}
\usepackage{enumitem}
\usepackage{apptools}
\usepackage{titlesec}
\usepackage{appendix}
\usepackage{xcolor}
\usepackage{dsfont}
\usepackage{authblk}

\AtAppendix{\counterwithin{lemma}{section}}

\theoremstyle{plain}

\newtheorem{theorem}{Theorem}
\newtheorem{corollary}[theorem]{Corollary}
\newtheorem{lemma}[theorem]{Lemma}

\newtheorem{remark}[theorem]{Remark}

\theoremstyle{definition}
\newtheorem{definition}[theorem]{Definition}

\newcommand{\RL}{\mathbb{R}}

\newcommand{\IPn} {
I_{P_n}
}
\newcommand{\IPu} {
I_{P^u}
}
\newcommand{\inn}[2] {
\langle #1 , #2 \rangle
}

\newcommand{\bergstrom} {Bergstr{\"o}m }

\title{Entropic versions of Bergstr{\"o}m's and Bonnesen's  inequalities}
\date{}

\author[1]{Matthieu Fradelizi}
\author[2]{Lampros Gavalakis}
\author[3]{Martin Rapaport}

\affil[1]{Univ Gustave Eiffel
                      }
                      
\affil[2]{University of Cambridge 
}

\affil[3]{
Carnegie Mellon University
                }

\begin{document}

\maketitle
 \footnotetext[1]{M.F. is with Univ Gustave Eiffel, Univ Paris Est Creteil, 
 CNRS, LAMA UMR8050 F-77447,
 Marne-la-Vall{\'e}e, France.\\
 email: matthieu.fradelizi@univ-eiffel.fr 
                      
L.G. is with the Department of Pure Mathematics and Mathematical Statistics, University of Cambridge,  UK and was supported by the Economic and Physical Sciences Research Council [Grant Ref:EP/Y028732/1].\\
email: lg560@cam.ac.uk  

M.R. is with the 
Department of Mathematical Sciences, Carnegie Mellon University, Pittsburgh, 15213, PA, United States. \\
email: mrapapor@andrew.cmu.edu }
 
\begin{abstract} 
We establish analogues of the \bergstrom and Bonnesen inequalities, related to determinants and volumes respectively, for the entropy power and for the Fisher information. The obtained inequalities strengthen the well-known convolution inequality for the Fisher information as well as the entropy power inequality in dimensions $d>1$, while they reduce to the former in $d=1$. Our results recover the original \bergstrom inequality and generalize a proof of Bergstr{\"o}m's inequality given by Dembo, Cover and Thomas. We characterize the equality case in our entropic Bonnesen inequality. 
\end{abstract}
\section{Introduction}

The analogy between entropy and volume has been observed as early as in the work of Costa and Cover \cite{costacover}, where the connection between the Entropy Power Inequality (EPI) and the Brunn-Minkowski inequality is examined (see \cite{dembocoverthomas} for an extensive review and a common proof of these inequalities). 

On the other hand, a number of determinant inequalities may be proved using entropy \cite{dembocoverthomas, cover:book}.

First, let us recall the \bergstrom and Bonnesen  inequalities in convex geometry, which will serve as the starting point for the analogies in information theory. The classical matrix form of the \bergstrom inequality is the following: 
\begin{theorem} \label{BergstromOrg} \cite{Bergstrom}
Let $A$ and $B$ be two $n \times n$ positive definite real symmetric matrices, and denote by $A_i$ and $B_i$ the two $(n - 1) \times (n - 1)$ matrices resulting from $A$ and $B$ by deleting the $i$-th row and the $i$-th column. Then we have
\[
\frac{\det(A + B)}{\det(A_i + B_i)} \geq \frac{\det(A)}{\det(A_i)} + \frac{\det(B)}{\det(B_i)},
\]
for every $i \in \{1, \dots, n\}$.
\end{theorem}
The above theorem is, by setting $A = \lambda S$ and $B = (1-\lambda)T$, equivalent to the statement that 
$\frac{\det{(A)}}{\det{(A_i)}}$ is concave in $A$.
A proof of this statement using entropy was given in \cite{dembocoverthomas}, by considering the entropy of Gaussian random vectors with covariances $A$ and $B$ and observing that conditioning reduces entropy. 

The first question that motivated the current work is whether an entropic analogue of Theorem \ref{BergstromOrg} can be obtained. That is, can we derive an inequality that holds for all random variables and generalizes the inequality obtained in the proof of \cite[Theorem 30]{dembocoverthomas} for Gaussians? Our first main result, Theorem \ref{B2} provides such an inequality (see \eqref{entropyPowerBonnesen}). 

Let us also mention here Ky-Fan's inequality, which extends \bergstrom inequality:

\begin{theorem} \cite{Fan} \label{fanlabel}
Let $A$ and $B$ be two $n \times n$ positive definite real symmetric matrices, and denote by $A_{(k)}$ and $B_{(k)}$ the principal $(n-k) \times (n-k)$ matrices of $A$ and $B$ obtained by taking the first $n-k$ rows and $n-k$ columns from $A$ and $B$, respectively. Then we have
\[
\left( \frac{\det(A + B)}{\det(A_{(k)} + B_{(k)})} \right)^{\frac{1}{k}} \geq 
\left( \frac{\det(A)}{\det(A_{(k)}}\right)^{\frac{1}{k}} + 
\left( \frac{\det(B)}{\det(B_{(k)}}\right)^{\frac{1}{k}}
,
\]
for every $k \in \{1, \dots, n - 1\}$.
\end{theorem} 

\bergstrom inequality implies that if $\det(A_i) = \det(B_i)$ for some $i \in \{1, \dots, n\}$, then 
\begin{equation} \label{bergassumption}
\det(\lambda A + (1-\lambda)B)) \geq \lambda \det(A) + (1-\lambda) \det(B).
\end{equation}
That is, the well known concavity of $A \mapsto \det(A)^{\frac{1}{n}}$, or in equivalent terms, the inequality  $\det(\lambda A + (1-\lambda)B)^{\frac{1}{n}} \geq \lambda \det(A)^{\frac{1}{n}} + (1-\lambda) \det(B)^{\frac{1}{n}}$, can be improved to \eqref{bergassumption}, which fails in general, under the assumption that the matrices obtained by removing some column and the corresponding row have equal determinants.

Another theorem due to Bonnesen implies linear refinements of the Brunn-Minkowski inequality and can be seen as a weaker volume-analogue of \eqref{bergassumption}. Here and in what follows, we denote with $|K|_n$ the volume of a set $K$ in $\mathbb{R}^n$.

\begin{theorem}\label{bonnessenintro}
\cite{bonnesenp}
Let $A$ and $B$ be convex bodies in $\mathbb{R}^n$ and let $\theta \in S^{n-1}$ then 
$$
\frac{|A+B|_n}{\left(|P_{\theta^{\perp}}A|_{n-1}^{\frac{1}{n-1}}+|P_{\theta^{\perp}} B|_{n-1}^{\frac{1}{n-1}}\right)^{n-1}}\ge  \frac{|A|_{n}}{|P_{\theta^{\perp}} A|_{n-1}}+\frac{|B|_{n}}{|P_{\theta^{\perp}}B|_{n-1}},
$$
where \( P_\theta^\perp A \) denotes the orthogonal projection of \( A \) onto the hyperplane with normal vector $\theta$.
\end{theorem}
Let $K$ and $L$ be two convex bodies in $\mathbb{R}^n$ such that $\left| P_{\theta^{\perp}}K \right|_{n-1} = \left| P_{\theta^{\perp}} L \right|_{n-1}$ for some $\theta \in S^{n-1}$. Then, applying Bonnesen's inequality above to $A=(1-\lambda)K$ and $B=\lambda L$ we deduce that 
\begin{equation} \label{volumeimprov}
|(1-\lambda)K+\lambda L|_n\ge (1-\lambda)|K|_{n}+\lambda|L|_{n},
\end{equation}
which is the linear improvement of Brunn-Minkowski's inequality. 
It would be natural to conjecture that a stronger Bonnesen's inequality in the form
\begin{equation} \label{falseconj} 
\frac{|A+B|_n}{|P_{\theta^{\perp}}(A+ B)|_{n-1}}\ge  \frac{|A|_{n}}{|P_{\theta^{\perp}} A|_{n-1}}+\frac{|B|_{n}}{|P_{\theta^{\perp}}B|_{n-1}}
\end{equation}
holds for any convex bodies $A$ and $B$ in $\mathbb{R}^n$ but, in \cite{FMMZ}, it was disproved  in general for $n\ge3$ and proved for $A$ and $B$ being zonoids in $\mathbb{R}^3$.

Connections between this type of inequalities in convex geometry and matrix inequalities are studied in \cite{saorin}.

In view of the improvements \eqref{bergassumption} and \eqref{volumeimprov} and the close connection between the Brunn-Minkowski and EPI, it is natural to ask under which assumptions one can improve the EPI in an analogous manner. 
We recall here, that one of the equivalent formulations of the EPI states that, if $X,Y$ are independent random vectors in $\mathbb{R}^n$, then  
$$
e^{\frac{2}{n}h(X+Y)} \geq e^{\frac{2}{n}h(X)} + e^{\frac{2}{n}h(Y)}.
$$

In Corollary \ref{bonnesen} below, we show that if two $n$-dimensional random vectors have some $n-1$-dimensional marginal with the same entropy, then the EPI can be improved by removing the $\frac{1}{n}$ factor from the exponent. 

Moreover, we characterize the equality case, which turns out to be if and only if $X,Y$ are Gaussians with the same covariance matrix up to the last element of the diagonal. This is in complete analogy with the equality case in \eqref{bergassumption}.

As an analogue to the Fisher information inequality 
$$
I(X+Y)^{-1} \geq I(X)^{-1} + I(Y)^{-1},
$$
Dembo, Cover, Thomas \cite{dembocoverthomas} asked whether the inequality 
\begin{equation} \label{dctC}
\frac{|K+L|}{|\partial(K+L)|} \geq \frac{|K|}{|\partial K|}+ \frac{|L|}{|\partial L|}
\end{equation}
holds true for every convex bodies $K$ and $L$ in $\mathbb{R}^n$. More generally, Vitali Milman  asked if 
\begin{equation} \label{mixedvol}
\frac{V_k(K+L)}{V_{k-1}(K+L)} \geq \frac{V_k(K)}{V_{k-1}(K)}+ \frac{V_k(L)}{V_{k-1}(L)},
\end{equation}
holds true, where $V_k$ denote the mixed volumes. In \cite{fmg}, the latter was shown to hold true if and only if $k=1,2$, implying that \eqref{dctC} holds true only if $n=2$. 

In view of the analogy between entropy and volume, our entropic \bergstrom inequality \eqref{entropyPowerBonnesen} also serves as an entropic analogue of \eqref{dctC}. As we will see, this is always true, even without convexity assumptions. Moreover, it is in the analogous form of the volume inequality \eqref{falseconj}, which is not true, i.e. the entropic analogue is always true even though the volume version fails in general. This should not be surprising, since entropy usually serves as a more flexible analogue of volume. This is often the case in discrete settings as well; it was recently highlighted by the breakthrough proof of Marton's conjecture by Gowers, Green, Manners and Tao \cite{marton1, marton2, revisited}, where the observation that entropy behaves well under group homomorphisms, in contrast to cardinality, turned out to be crucial. 

In the same spirit, the following  related result addresses convex bodies \( K \) and \( L \) possessing equal maximal hyperplane sections in a given direction. 
\begin{theorem} \cite{bonnesens}
Let $K$ and $L$ be two convex bodies in $\mathbb{R}^{n}$ such that  
\[
\sup_r \left|K \cap (\theta^\perp + r\theta)\right|_{n-1} = \sup_s \left|L \cap (\theta^\perp + s\theta)\right|_{n-1}
\]
for some \( \theta \in \mathbb{S}^{n-1} \), then for every \( \lambda \in (0, 1) \) we have:
\[
 \left|(1 - \lambda)K + \lambda L \right|_{n} \geq (1 - \lambda) |K|_{n} + \lambda |L|_{n}.
\]
\end{theorem}
We refer to the latter result as the Bonnesen inequality for sections (rather than projections which refers to Theorem \ref{bonnessenintro}). \\

As a corollary of our main result we also obtain an analogous improvement of the isoperimetric inequality for entropies. 

Finally, in view of the different forms of the \bergstrom and Bonnesen inequalities, we establish a different analogue for the Fisher information. To that end, in Section \ref{fisherinfosec} we define a conditional version of the Fisher information and prove an inequality, which resembles Bergstr{\"o}m's inequality. This turns out to be stronger than (in the sense that it implies) the convolution inequality for Fisher information (sometimes referred to as Blachman-Stam inequality). We do  not know whether this is stronger than (in that it implies) our entropic \bergstrom inequality.

\vspace{0.5 cm}

{\bf Notation.} We write capital letters $X,Y$ for random variables (resp. vectors) and small letters $x,y$ for specific realizations of these. If $X = (X_1,\ldots, X_n) \in \mathbb{R}^n$ we write $X^{n-1} := (X_1,\ldots,X_{n-1})$ to denote the first $n-1$ coordinates. When it is not clear from the context, we will write $h_n(X)$ and $N_n(X)$ for the entropy and entropy power respectively, to emphasize that the integral in the definition of entropy (see Section \ref{entropysec}) is with respect to the Lebesgue measure in $\mathbb{R}^n$.

\section{Entropy inequalities} \label{entropysec}

Recall that the differential entropy $h(X)$ and the entropy power $N(X)$ of a random vector $X$ in $\mathbb{R}^n$ with density $f$ are defined as  
\[
h(X)=-\int_{\mathbb{R}^n} f(x)\log f(x) dx\quad\hbox{and}\quad N(X)=e^{\frac{2}{n}h(X)}.
\]

Let \( X \) and \( Y \) be two independent random variables in \( \mathbb{R}^{d} \). The EPI states that \[ N(X+Y) \geq N(X) + N(Y). \]


In the proof of our main result we are going to need a conditional form of the EPI, which seems to be known (see e.g. \cite{conditionalquantum} where quantum versions are of interest). Nevertheless we include the proof, as we were not able to find it explicitly in the literature. 

Let $X,Z,Y$ be three random variables/vectors with densities. We say that $X\rightarrow Z \rightarrow Y$ form a Markov chain if,  $X$ and $Y$ are conditionally independent given $Z$, i.e. for a.e. $z,$ $f(x,y|z) = f(x|z)f(y|z),$ where $f(x,y,z)$ is the joint density of $X,Y,Z$ and for any $z,$ $f(x,y|z) = \frac{f(x,y,z)}{f(z)} = \frac{f(x,y,z)}{\int{f(u,v,z)dudv}}$ is the conditional density of $(X,Y)$ given $Z = z,$ defined for a.e. $z$, and analogously $f(x|z) =\frac{f(x,z)}{f(z)}=\frac{\int f(x,v,z)dv}{\int{f(u,v,z)dudv}}$.
\begin{lemma} \label{condEPI}
Suppose $X,Y \in \mathbb{R}^n$, and $Z$ are random  with values in some space $\Omega$,  such that $X\rightarrow Z \rightarrow Y$ form a Markov chain and given $Z,$ $X$ and $Y$ have conditional densities on $\mathbb{R}^n$. Then 
$$N(X+Y|Z) \geq N(X|Z) + N(Y|Z),$$
where for any random vectors $U,V$ such that $U$ has a conditional density given $V$ in $\mathbb{R}^n,$ $N(U|V) := e^{\frac{2}{n}h(U|V)}$.
Moreover, there is equality, if and only if for $Z$-almost every $z,$ the conditional densities given $Z=z$ of $X$ and $Y$ are Gaussian with proportional covariances.  
\end{lemma}
\begin{proof}

Let us consider 
\[
\tilde{X} = \frac{X}{\sqrt{N(X|Z)}} \quad \text{and} \quad \tilde{Y} = \frac{Y}{\sqrt{N(Y|Z)}}
\]
and observe that by scaling of entropy power
${N}(t X) = t^2 {N}(X),$ we have 
\begin{equation} \label{condH0}
N(\tilde{X}|Z) = N(\tilde{Y}|Z) = 1.
\end{equation}
Set 
\[
\lambda = \frac{N(X|Z)}{N(X|Z) + N(Y|Z)} 
\]
and for shorthand let $t:= N(X|Z) + N(Y|Z)$, i.e. $ \lambda = \frac{N(X|Z)}{t}$.
We have
\[
{N}(X + Y|Z) = {N} \left(\sqrt{t}( \sqrt{\lambda} \tilde{X} + \sqrt{1 - \lambda} \tilde{Y} ) |Z\right) = t {N}( \sqrt{\lambda} \tilde{X} + \sqrt{1 - \lambda} \tilde{Y}  |Z).
\]
But, since $\tilde{X}$ and $\tilde{Y}$ are conditionally independent given $Z$, it follows from the EPI, that for almost every $z,$
\begin{equation} \label{fixedzineq}
h(\sqrt{\lambda}{\tilde{X}} + \sqrt{(1-\lambda)}{\tilde{Y}}|Z=z) \geq \lambda h(\tilde{X}|Z=z) + (1-\lambda)h(\tilde{Y}|Z=z).
\end{equation}
Using the expression \eqref{condHaverage}, we get
\begin{equation} \label{condEPIlog}
h(\sqrt{\lambda}{\tilde{X}} + \sqrt{(1-\lambda)}{\tilde{Y}}|Z) \geq \lambda h(\tilde{X}|Z) + (1-\lambda)h(\tilde{Y}|Z),
\end{equation}
and using \eqref{condH0}, we have
${N}( \sqrt{\lambda} \tilde{X} + \sqrt{1 - \lambda} \tilde{Y}  |Z) \geq 1$. Therefore 
\[
{N}(X + Y|Z) \geq  t = {N}(X|Z)+ {N}(Y|Z).
\]

Moreover, if there is equality in the statement, then inequality \eqref{fixedzineq} should be equality for almost every $z$ (otherwise, if it were strict inequality in a set of positive measure, \eqref{condEPIlog} would be strict inequality).

From the equality case in this form of the EPI \cite{dembocoverthomas}, $\tilde{X}|_{Z=z}$ and $\tilde{Y}|_{Z=z}$ are Gaussian with the same covariance matrix. By the definitions of $\tilde{X}$ and $\tilde{Y}$, $X,Y$ are also Gaussian having as covariance matrix each a possibly different multiple of the common covariance of $\tilde{X}$ and $\tilde{Y}$.
\end{proof}

\begin{theorem} [Entropy analogue of Bergstr{\"o}m's inequality] \label{B2}
Let \(X = (X_{1}, \ldots, X_{n})\) and \(Y = (Y_{1}, \ldots, Y_{n})\) be two independent random vectors in \(\mathbb{R}^{n}\). Let \(X^{n-1} = (X_{1}, \ldots, X_{n-1})\) and \(Y^{n-1} = (Y_{1}, \ldots, Y_{n-1})\). Then, the following inequalities hold true and are equivalent: 
\begin{equation} \label{conditionalHBonnesen}
e^{2h(\sqrt{1-\lambda}X_n+\sqrt{\lambda}Y_n|\sqrt{1-\lambda}X^{n-1}+\sqrt{\lambda}Y^{n-1})} \geq (1-\lambda)e^{2h(X_n|X^{n-1})} + \lambda e^{2h(Y_n|Y^{n-1})}
\end{equation}
and
\begin{equation} \label{entropyPowerBonnesen}
\frac{N(X+Y)^n}{ N_{n-1}({X^{n-1}+Y^{n-1})^{n-1}}} \geq \frac{N(X)^n}{N_{n-1}(X^{n-1})^{n-1}} + \frac{N(Y)^n}{N_{n-1}(Y^{n-1})^{n-1}}.
\end{equation}

\end{theorem}

\begin{proof}
Inequality \eqref{conditionalHBonnesen} can be seen to be equivalent to 
$$
e^{2h(X_n+Y_n|X^{n-1}+Y^{n-1})} \geq e^{2h(X_n|X^{n-1})} +  e^{2h(Y_n|Y^{n-1})}
$$
by setting $\tilde{X} = \sqrt{1-\lambda}X$ (and analogously for $Y$) and scaling. 

The latter is equivalent to \eqref{entropyPowerBonnesen} by the definition of entropy power and since by the chain rule for differential entropy, 
$$
h\left(  X + Y \right) - h\left(  X^{n-1} +  Y^{n-1} \right) = h\left(X_n+Y_n \bigl| X^{n-1} +  Y^{n-1}
\right).
$$

We prove \eqref{conditionalHBonnesen}. Since conditioning reduces entropy,
\begin{align*}
&h\left(\sqrt{1-\lambda}X_n + \sqrt{\lambda} Y_n \mid \sqrt{1-\lambda}X^{n-1} + \sqrt{\lambda}Y^{n-1}\right) \\
&\geq h\left(\sqrt{1-\lambda}X_n + \sqrt{\lambda}Y_n \mid \sqrt{1-\lambda}X^{n-1} + \sqrt{\lambda}Y^{n-1},  X^{n-1}, Y^{n-1}\right) 
=h\left(\sqrt{1-\lambda}X_n + \sqrt{\lambda}Y_n \mid X^{n-1}, Y^{n-1}\right).
\end{align*}
Thus,
\begin{align} \nonumber
e^{2h(\sqrt{1-\lambda} X_n + \sqrt{\lambda} Y_n \mid \sqrt{1-\lambda}X^{n-1} + \sqrt{\lambda} Y^{n-1})} &\geq e^{2 h\left(\sqrt{1-\lambda}X_n + \sqrt{\lambda}Y_n \mid X^{n-1}, Y^{n-1}\right)} \\ \label{condEPIuse}
&\geq (1-\lambda) e^{2h(X_n \mid X^{n-1},Y^{n-1})}+\lambda e^{2h(Y_n\mid X^{n-1},Y^{n-1})} \\ \nonumber
&= (1-\lambda) e^{2h(X_n \mid X^{n-1})}+\lambda e^{2h(Y_n\mid Y^{n-1})},
\end{align}
where in \eqref{condEPIuse} we have used the conditional EPI, Lemma \ref{condEPI}, with $Z = (X^{n-1},Y^{n-1})$, after noting that by the independence of $X$ and $Y$, $X_n$ and $Y_n$ are conditionally independent given $(X^{n-1},Y^{n-1})$.

\end{proof}

A few remarks are in order: 
\begin{remark}
\begin{enumerate}

\item It can be seen from the proof of Theorem \ref{B2} that the following two equivalent forms also hold true:   
\begin{equation} 
e^{2h(X_n+Y_n|X^{n-1}+Y^{n-1})} \geq e^{2h(X_n|X^{n-1})} +  e^{2h(Y_n|Y^{n-1})}
\end{equation}
and 
\begin{equation} \label{lambdaform}
\frac{N\left(\sqrt{\lambda}X+\sqrt{1-\lambda}Y\right)^n}{ N_{n-1}\left(\sqrt{\lambda} X^{n-1}  + \sqrt{1-\lambda}Y^{n-1}\right)^{n-1}
}
\geq \lambda\frac{N(X)^n}{N_{n-1}(X^{n-1})^{n-1}} + (1-\lambda)\frac{N(Y)^n}{N_{n-1}(Y^{n-1})^{n-1}}.
\end{equation}

    \item Inequality \eqref{entropyPowerBonnesen} is stronger than the EPI, $N(X+Y) \geq N(X) + N(Y)$, for $n >1$, in the sense that the latter can be deduced from the former via the following inductive argument: 

Assume \eqref{entropyPowerBonnesen} and the EPI for $n=1$. By the EPI for $n-1$, $N_{n-1}(X^{n-1} + Y^{n-1}) \geq N_{n-1}(X^{n-1}) + N_{n-1}(Y^{n-1}),$ and therefore \eqref{entropyPowerBonnesen} implies 
$$
N(X+Y)^n \geq \Bigl(\frac{N_{n-1}(X^{n-1})+N_{n-1}(Y^{n-1})}{N_{n-1}(X^{n-1})}\Bigr)^{n-1} N(X)^n + \Bigl(\frac{N_{n-1}(X^{n-1})+N_{n-1}(Y^{n-1})}{N_{n-1}(Y^{n-1})}\Bigr)^{n-1} N(Y)^n.
$$
    Letting $\mu = \frac{N_{n-1}(X^{n-1})}{N_{n-1}(X^{n-1})+N_{n-1}(Y^{n-1})}$,
    this reads  
    $$
    N(X+Y)^n \geq \mu \frac{N(X)^n}{\mu^n} + (1-\mu)\frac{N(Y)^n}{(1-\mu)^n} \geq (N(X) + N(Y))^n,
    $$
by convexity of $x \mapsto x^n, x\geq 0$, the EPI follows. 

\item The steps in the proof of Theorem \ref{B2} are generalizations of the proof of \cite{dembocoverthomas}, which restricted to Gaussian random variables. Taking $X,Y$ Gaussian in our result, we recover Bergstr{\"o}m's inequality for determinants.

\item Although we state Theorem \ref{B2} for simplicity using $N_{n-1}(X^{n-1})$ and $N_{n-1}(Y^{n-1})$, it can be generalized, via a change of axes, in the sense that the same conclusion holds with $N_{n-1}(AX)$ and $N_{n-1}(AY)$ instead, for any projection $A: \RL^n \to \RL^{n-1}$.

\item The same proof also works to obtain the following entropic analogue of Ky-Fan's Theorem (Theorem \ref{fanlabel}), which generalizes \eqref{conditionalHBonnesen}: for $X=(X_1,\dots,X_n)$ and $I \subset \{1,\ldots,n\}$, write $X_I = \{X_i\}_{i \in I}$. For $I$ with $|I| = k$, we have
\begin{equation} 
e^{\frac{2}{k}h(\sqrt{1-\lambda}X_I+\sqrt{\lambda}Y_I|\sqrt{1-\lambda}X_{I^c}+\sqrt{\lambda}Y_{I^c})} \geq (1-\lambda)e^{\frac{2}{k}h(X_I|X_{I^c})} + \lambda e^{\frac{2}{k}h(Y_I|Y_{I^c})}.
\end{equation}
\end{enumerate}

\end{remark}

Theorem \ref{B2} implies an entropic version of Bonnesen's inequality \ref{bonnessenintro}, Corollary \ref{bonnesen} below. In view of the concavity of $x \mapsto x^{\frac{1}{n}}, x \geq 0$, it shows that, under the assumption that some $n-1$--dimensional subvectors of $X$ and $Y$ have the same entropy, the EPI can be improved. Again, although we state it for simplicity under the assumption \(h(X^{n-1}) = h(Y^{n-1})\), the latter can be relaxed via a change of axes to \(h(AX) = h(AY)\) for some projection $A: \RL^n \to \RL^{n-1}$.

\begin{corollary} [Entropy analogue of Bonnesen's inequality] \label{bonnesen}
Let \(X = (X_{1}, \ldots, X_{n})\) and \(Y = (Y_{1}, \ldots, Y_{n})\) be two independent random vectors in \(\mathbb{R}^{n}\). Let \(X^{n-1} = (X_{1}, \ldots, X_{n-1})\) and \(Y^{n-1} = (Y_{1}, \ldots, Y_{n-1})\) and suppose that \(h(X^{n-1}) = h(Y^{n-1})\). Then, 

\begin{equation}
\label{result1}
e^{2h(\sqrt{1-\lambda}X + \sqrt{\lambda}Y)} \geq (1-\lambda)e^{2h(X)} + \lambda e^{2h(Y)},\hspace{0.1cm}
\text{for all}\hspace{0.1cm} \lambda\in [0,1].
\end{equation}
Moreover, there is equality if and only if $X$ and $Y$ are Gaussians having the same covariance matrix except the last element of the diagonal. 
\end{corollary}

\begin{proof}
By the EPI 
\begin{equation} \label{useofEPI}
N_{n-1}\left(\sqrt{\lambda} X^{n-1}  + \sqrt{(1-\lambda)}Y^{n-1}\right) \geq \lambda N_{n-1}(X^{n-1}) + (1-\lambda)N_{n-1}(Y^{n-1})=N_{n-1}(X^{n-1}).
\end{equation}
Combining with \eqref{lambdaform} the claim follows by simplifying since $N_{n-1}(X^{n-1}) = N_{n-1}(Y^{n-1})$.

For the equality case, note that if there is equality in \eqref{result1}, all inequalities used in its proof have to be equalities. But equality in \eqref{useofEPI} implies that $X^{n-1}$ and $Y^{n-1}$ are Gaussians with proportional covariances. Since, by assumption, $h(X^{n-1}) = h(Y^{n-1})$, the covariances of $X^{n-1}$ and $Y^{n-1}$ have in addition the same determinant. Therefore, they have to be equal to each other. 

Moreover, we must have equality in \eqref{condEPIuse}. By the equality case in Lemma \ref{condEPI} the conditional distributions $X_n$ and $Y_n$ given $X^{n-1}$ and $Y^{n-1}$ respectively, are Gaussian. Therefore, $X$ and $Y$ are both Gaussian. 

Now we argue that the elements of the last rows (respectively columns) of the covariances are the same. 
So, suppose \( X \) and \( Y \) are Gaussian random vectors with respective covariance matrices \( \Sigma_1^n \) and \( \Sigma_2^n \). The equality case in \eqref{result1} translates to
\begin{equation}\label{equaldeterminants}
\left| (1 - \lambda) \Sigma_{1}^{n} + \lambda \Sigma_{2}^{n} \right| = (1 - \lambda) \left| \Sigma_{1}^{n} \right| + \lambda \left| \Sigma_{2}^{n} \right|.
\end{equation}
As we have already argued, 
$\left| \Sigma_1^{n-1} \right| = \left| \Sigma_2^{n-1} \right|,$
where $\Sigma_1^{n-1}$ is the sub-matrix obtained after removing the last row and column of $\Sigma_1^n$. Without loss of generality, we can assume that
$
\Sigma_1^{n-1} = \Sigma_2^{n-1} = I_{n-1}
$, where $I_{n-1}$ is the $(n-1)$ identity matrix.
Thus, \( \Sigma_{1}^{n} = \begin{pmatrix} I_{n-1} & v \\ v^T & a_n \end{pmatrix} \)  where \( v = \begin{pmatrix} a_1 \\ a_2 \\ \vdots \\ a_{n-1} \end{pmatrix} \) is a column vector of size \( (n-1) \), and the matrices $\Sigma_{2}^{n}, (1-\lambda)\Sigma_{1}^{n}+\lambda\Sigma_{2}^{n}$ have  analogous forms. For the determinant we have  
\[
\left| \Sigma_1^n \right| = \left|I_{n-1} \right| \cdot \left| a_n - v^T I_{n-1}^{-1} v\right| = a_n - \sum_{i=1}^{n-1} a_i^2
\]
and analogously for $\Sigma_2^n$ and $(1-\lambda)\Sigma_{1}^{n}+\lambda\Sigma_{2}^{n}$. Therefore, writing $b_i$ for the first $n-1$ elements of the last row (respectively column) of $\Sigma_2^n$, \eqref{equaldeterminants} corresponds to 

\[
(1 - \lambda) \sum_{i=1}^{n-1} a_i^2 + \lambda \sum_{i=1}^{n-1} b_i^2 = \sum_{i=1}^{n-1} \left( (1 - \lambda) a_i + \lambda b_i \right)^2.
\]
Thus, we conclude that equality occurs only if \( a_i = b_i \) for all \( i = 1, \dots, n-1 \).

\end{proof}

\begin{remark}
Notice that by defining the map:
\begin{equation} \label{Qmap}
f: \lambda \mapsto \frac{N\left(\sqrt{\lambda}X+\sqrt{1-\lambda}Y\right)^{n}}{\left(N_{n-1}(\sqrt{\lambda}X^{n-1} +\sqrt{(1-\lambda)}Y^{n-1})\right)^{n-1} }, 
\end{equation}

Theorem \ref{B2} states that
\[
f(\lambda) = f(\lambda \cdot 1 + (1-\lambda) \cdot 0) \geq \lambda f(1) + (1-\lambda) f(0).\]

Ball, Nayar and Tkocz \cite{bnt} conjectured that the map $\lambda \mapsto h(\sqrt{1-\lambda}X+\sqrt{\lambda}Y)$ is concave for i.i.d. log-concave $X,Y$ in dimension 1. 
A partial result in a different direction was given in \cite{eskegavalakis}, where it shown that this entropy-map is concave for Gaussian (scale) mixtures. 

In view of the above conjecture and the partial answer for Gaussian mixtures it is natural to ask for which class of random variables (if any) is the map defined in \eqref{Qmap} concave. 

Note that, in contrast to these works, when restricting to $d=1$, we are asking for concavity of the entropy power, rather than the entropy itself, which is stronger, since concavity implies log-concavity. 

When one of the two random vectors is Gaussian, by homogeneity, concavity of the map defined above is equivalent to concavity of 
$$
f: \lambda \mapsto \frac{N\left(X+\sqrt{\lambda}Z\right)^{n}}{\left(N_{n-1}(X^{n-1}+\sqrt{\lambda}Z^{n-1})\right)^{n-1} }, 
$$
analogously to \cite[Proposition 2.2]{fradelizimarsiglietti}.
The latter is known to hold true in dimension $1$ \cite{costa}.

\end{remark}
 
The isoperimetric inequality for entropies \cite[Theorem 16]{dembocoverthomas} asserts that 
\begin{equation} \label{usualIso}
    I(X)N(X) \geq 2 \pi e n
\end{equation} 
and is a consequence of the EPI. By our strengthened EPI, we obtain a sharpening of the isoperimetric inequality, Corollary \ref{isoperimetricCor} below. 

To see that \eqref{isoperimetric} is indeed stronger than the usual isoperimetric inequality for entropies, let $a = \frac{N(X^{n-1})}{N(X)}$ in \eqref{isoperimetric} and note that the arithmetic-geometric mean inequality implies $\frac{1}{n}a^{n-1} + \frac{n-1}{n}\frac{1}{a} \geq 1$. The isoperimetric inequality \eqref{usualIso} then follows.

\begin{corollary} \label{isoperimetricCor}
Let \( X \) be a random vector in \( \mathbb{R}^n \). Then, the following inequality is satisfied:

\begin{equation} \label{isoperimetric}
I(X) N(X) \geq 2\pi e \left( \left( \frac{N(X^{n-1})}{N(X)} \right)^{n-1} + (n-1) \frac{N(X)}{N_{n-1}(X^{n-1})} \right).
\end{equation}

\end{corollary}

\begin{proof}

Let \( X \) be a random vector in \( \mathbb{R}^n \) and \( Z \) be a standard Gaussian vector independent of \( X \). By Theorem~\ref{B2},

\[
\frac{N(X + \sqrt{t} Z)^n}{ (N_{n-1}(X^{n-1}) + t N_{n-1}(Z^{n-1}))^{n-1}} \geq \frac{N(X)^n}{N_{n-1}(X^{n-1})^{n-1}} + t \frac{N(Z)^n}{N_{n-1}(Z^{n-1})^{n-1}}.
\]
Noting that, if \( Z \) is an \( n \)-dimensional standard normal, \( N(Z) = 2\pi e \), by a first-order Taylor expansion of the convex function $t \mapsto  (N_{n-1}(X^{n-1}) + t N_{n-1}(Z^{n-1}))^{n-1}, t \geq 0$, we obtain
\[
N(X + \sqrt{t} Z)^n \geq \left(N_{n-1}(X^{n-1})^{n-1} + 2\pi e t (n-1) N_{n-1}(X^{n-1})^{n-2})\right) \left(\frac{N(X)^n}{N_{n-1}(X^{n-1})^{n-1}} + 2\pi e t\right).
\]
Therefore,
\begin{equation} \label{taylor}
N(X + \sqrt{t} Z)^n \geq N(X)^n + 2\pi e t \left(N_{n-1}(X^{n-1})^{n-1}+ (n-1) \frac{N(X)^n}{N_{n-1}(X^{n-1})}\right) + o(t^2).
\end{equation}
\noindent
This implies, as \( t \to 0 \), 

\[
\lim_{t\to 0} \frac{N(X+\sqrt{t}Z)^{n}- N(X)^{n}}{t} \geq 2\pi e  \left(N_{n-1}(X^{n-1})^{n-1}+ (n-1) \frac{N(X)^n}{N_{n-1}(X^{n-1})}\right).
\]

We recall de Bruijn's identity \cite[Theorem 14]{dembocoverthomas}, which states that if $Z$ is a standard normal, then
\begin{equation} \label{debruijn}
 \frac{d}{dt} h(X + \sqrt{t} Z) = \frac{1}{2} I(X + \sqrt{t} Z).
\end{equation}

Therefore the term on the left-hand side of \eqref{taylor} corresponds to \( \frac{d}{dt} N(X + \sqrt{t} Z)^n  \left. \right|_{t=0} \), which by \eqref{debruijn} equals \( I(X) N(X)^n  \) and the claimed inequality follows.

\end{proof}

\section{Fisher information inequalities} \label{fisherinfosec}

Let \( X \) be a random vector in \( \mathbb{R}^n \) with probability density function \( f \). Then its Fisher information is defined as 
\[
I(X) = \int_{\mathbb{R}^n} \frac{\left| \nabla f(x) \right|^{2}}{f(x)} \, dx.
\]
Slightly abusing the notation, we also write $I(f) = I(X)$ when $X$ has density $f$.

The following inequality, sometimes referred to as Blachman-Stam or Fisher information inequality, asserts that for any pair of independent random vectors $ X$ and $ Y$ in $ \mathbb{R}^n $, 

$$ I(X+Y)^{-1}\geq I(X)^{-1}+I(Y)^{-1} .$$

In what follows, if $X$ has density $f$ on $\mathbb{R}^n$, we write 
\begin{equation}
    \IPn(X) =  \int_{\mathbb{R}^n} \frac{\langle \nabla f, e_n \rangle^2}{f(x)} \, dx,
\end{equation}
and we call $\IPn$ the {\em projective Fisher information} of $X$.

\begin{definition} [{\em Conditional Fisher information}]
For two random vectors $X \in \mathbb{R}^n,Y \in \mathbb{R}^m$, such that the conditional density of $X$ given $Y$, say $f_{X|Y}(\cdot|\cdot)$, exists a.s. and is absolutely continuous on $\mathbb{R}^k$ for some $1 \leq k \leq n$, we define the conditional Fisher Information of $X$ given $Y$ as the expected Fisher information of the conditional density, that is 
$$
I(X|Y) := \int_{\mathbb{R}^m}{ f_Y(y) I\bigl(f_{X|Y}(\cdot | y)\bigr) dy}.
$$
\end{definition}

Two properties of the projective and conditional Fisher information of a random vector $X = (X_1,\ldots,X_n) \in \mathbb{R}^n$  are straightforward from the definitions: 
\begin{enumerate}
    \item 
    \begin{equation} \label{projeqcond} 
        \IPn (X) = I(X_n|X^{n-1}).
    \end{equation}
    This is because 
    \begin{align*}
       \IPn (X) &=  \int_{\mathbb{R}^n}\frac{\langle \nabla f, e_n \rangle ^2}{f(x)}  dx \\
       &= \int_{\RL^n}{\frac{f(x_1,\ldots,x_{n-1})\bigl(\frac{d}{dx_n}f(x_n|x_1,\ldots,x_{n-1})\bigr)^2}{f(x_n|x_1,\ldots,x_{n-1})}}dx \\
       &= \int_{\RL^{n-1}}f(x_1,\ldots,x_{n-1})\int_{\RL}{\frac{\bigl(\frac{d}{dx_n}f(x_n|x_1,\ldots,x_{n-1})\bigr)^2}{f(x_n|x_1,\ldots,x_{n-1})}}dx_ndx^{n-1}.
    \end{align*}
    
    \item By \eqref{projeqcond} and since for any vector $u, \inn{u}{e_n}^2 \leq \|u\|_2^2$, we have, 
    \begin{equation}
        I(X_n|X^{n-1}) = \IPn(X) \leq I(X).
    \end{equation}
    
\end{enumerate}

Applying the Blachman-Stam inequality to a particular operator, we obtain the following conditional Fisher information inequality: 


\begin{theorem}
Let $X,Y$ be two random vectors in $\RL^{n}$ with finite Fisher informations. Then 
\begin{equation} \label{projInequality}
\IPn(X+Y)^{-1} \geq \IPn(X)^{-1} + \IPn(Y)^{-1},
\end{equation}
or, equivalently, 
\begin{equation}
    I(X_n+Y_n|X^{n-1}+ Y^{n-1})^{-1} \geq I(X_n|X^{n-1})^{-1} + I(Y_n|Y^{n-1})^{-1},
\end{equation}
provided that the conditional Fisher informations exist. 
\end{theorem}

\begin{proof}

Let $T_m: \mathbb{R}^n \to \mathbb{R}^n$ be the linear operator

\[
T_m = \mathrm{I}_{n} + \left( \frac{1}{m} - 1 \right) e_n e_n^T,
\]
where $\mathrm{I}_{n}$ denotes the identity matrix. 
Then, writing $f$ for the density of $X$,
$T_{m}(X)$ has density
\[
 f_{T_{m}}(x) :=\frac{1}{\lvert T \rvert} f(T_{m}^{-1} x), x \in \mathbb{R}^n .
\]
Thus, the Fisher information $I(T_{m}(X))$ is given by
\[
I(T_{m}(X)) = \int_{\mathbb{R}^{n}} \frac{\lvert \nabla f_{T_{m}}(x) \rvert^2}{f_{T_{m}(x)}} \, dx
= \int_{\mathbb{R}^{n}} \frac{\lvert (T^{-1})^T \nabla f(y) \rvert^2}{f(y)} \, dy
= \int_{\mathbb{R}^{n}} \frac{\sum_{i=1}^{n-1} (\frac{\partial}{\partial y_i} f)^2 + m^2 (\frac{\partial}{\partial y_n} f)^2}{f(y)} \, dy.
\]

Therefore, using the finiteness of the Fisher information of $X$,
\[
\lim_{m \to \infty} \frac{I(T_{m}(X))}{m^2} = \int_{\mathbb{R}^{n}} \frac{(\frac{\partial}{\partial y_n} f)^2}{f(y)} \, dy = \IPn(X).
\]

Analogously, $\lim_{m \to \infty} \frac{I(T_{m}(Y))}{m^2}  = \IPn(Y)$ and $\lim_{m \to \infty} \frac{I(T_{m}(X+Y))}{m^2}  = \IPn(X+Y)$.

Now the result follows by applying the Blachman-Stam inequality to the independent random variables $T_m(X)$ and $T_m(Y)$, after dividing by $m^2$ and letting $m$ tend to infinity.
\end{proof}

\begin{remark}
An examination of the proof shows that the same inequality can be obtained for the more general functional $I_{P^u} := \int_{\mathbb{R}^n}\frac{\langle \nabla f, u \rangle ^2}{f(x)}  dx, u \in \mathbb{S}^{n-1}.$ Using the useful identity 
\begin{equation} \label{intonsphere}
    \int_{\mathbb{S}^{n-1}}{\inn{u}{v}^2d\sigma(u)} = \frac{\|v\|^2}{n},
\end{equation}
where $d\sigma$ denotes the uniform measure on the sphere and which holds for any vector $v$, we can recover the Blachman-Stam inequality. Indeed, by Minkowski's inequality for $p=-1,$

\begin{align*}
I(X+Y) &= n\int_{\mathbb{S}^{n-1}} \IPu(X+Y)) \, d\sigma(u) \leq n\int_{\mathbb{\mathbb{S}}^{n-1}} \left( \IPu(X))^{-1} + \IPu(Y))^{-1} \right)^{-1} \, d\sigma(u) \\
&\leq \left( \left( n\int_{\mathbb{S}^{n-1}} \IPu(X)) \, d\sigma(u) \right)^{-1} + \left( n\int_{\mathbb{S}^{n-1}} \IPu(Y))\, d\sigma(u) \right)^{-1} \right)^{-1} \\
&= \Bigl(I(X)^{-1} + I(Y^{-1})\Bigr)^{-1}.
\end{align*}

To justify \eqref{intonsphere} note that the left-hand side is independent of the direction of $v$, i.e. it only depends on $\|v\|$. Therefore, by homogeneity it equals 
$$
\|v\|^2 \int_{\mathbb{S}^{n-1}}{\inn{u}{e_i}^2d\sigma(u)}
$$
for every $i=1,\ldots,n$ and thus
$$
\int_{\mathbb{S}^{n-1}}{\inn{u}{v}^2d\sigma(u)} = \frac{\|v\|^2}{n}\sum_{i=1}^n{\int_{\mathbb{S}^{n-1}}{\inn{u}{e_i}^2d\sigma(u)}} = \frac{\|v\|^2}{n}\int_{\mathbb{S}^{n-1}}\|u\|^2d\sigma(u) = \frac{\|v\|^2}{n}.
$$
\end{remark}

 \newpage 
\begin{appendices}
\appendixpage
\section{Basic properties of differential entropy}

In this appendix we summarize a few properties of entropy that we used in the proofs. 

We recall here that the differential entropy of a vector $X$ with density $f$ on $\mathbb{R}^n$ is 
$$
h(X) = h(f) := -\int_{\mathbb{R}^n}{f(x)\log{f(x)} dx} \quad \in [-\infty, \infty].
$$
A change of variables in the integral shows that the differential entropy satisfies, for any constant matrix $A \in \mathbb{R}^{n \times n}$, the scaling property
\begin{equation}
    h(AX) = h(X) + \log{|\det{(A)}|}.
\end{equation}

If $X \in \mathbb{R}^n,Y \in \mathbb{R}^m$ are two random variables (resp. vectors) with joint density $f_{X,Y}$ on $\mathbb{R}^{m+n}$ then the conditional density of $Y$ given $X$ exists a.s. and is given for any $y \in \mathbb{R}^m$ by 
$$
f_{X|Y}(x|y) := \frac{f_{X,Y}(x,y)}{f_Y(y)},
$$
where $f_Y(y) = \int{f_{X,Y}(x,y) dx}$ is the marginal of $Y$.
Then the joint and conditional entropies are defined respectively as 
$$
h(X,Y) := -\int_{\RL^{m+n}}{f(x,y)\log{f(x,y)}dxdy}
$$
and
$$
h(X|Y) := -\int_{\RL^{n+m}}{f_{X,Y}(x,y)\log{f_{X|Y}(x|y)}dxdy}.
$$
The conditional entropy can also be re-written as 
\begin{align} \nonumber
h(X|Y) &= -\int_{\RL^{n+m}}{f_{X,Y}(x,y)\log{f_{X|Y}(x|y)}dxdy} \\ \label{condHexpression}
&= -\int_{\RL^{m}} {f_Y(y)\int_{\RL^n}{f_{X|Y}(x|y)\log{f_{X|Y}(x|y)}dx}dy} \\ \label{condHaverage}
&= \int_{\RL^m}{f_Y(y)h(X|Y=y)}dy, 
\end{align}
where we denote, for $y \in \RL^m$, $h(X|Y=y) := h\bigl(f_{X|Y}(\cdot|y)\bigr).$ More generally, if $Y$ does not have a density, \eqref{condHaverage} can be replaced by 
$$
h(X|Y)=\mathbb{E}_Y{h(X|Y=y)} = \int{h(X|Y=y) d\mu_Y(y)}
$$
where $\mathbb{E}_Y$ denotes the expectation with respect to the law, $\mu_Y$, of $Y$.

The chain rule for entropy \cite{cover:book} is 
\begin{equation} \label{chainrule}
h(X,Y) = h(Y) + h(X|Y) = h(X) + h(Y|X).
\end{equation}

We also have that conditioning reduces entropy, i.e. 
\begin{equation} \label{conditioningreducesH}
h(X|Y) \leq h(X),
\end{equation}
since the difference between the two sides is \cite{cover:book}
\begin{equation}
    \int_{\mathbb{R}^{m+n}}{f_{X,Y}(x,y)\log{\frac{f_{X,Y}(x,y)}{f_{X}(x)f_Y(y)}}dxdy} \geq 0,
\end{equation}
by convexity of $t \to t\log{t}$ and Jensen's inequality (in fact the last expression is known as the mutual information $I(X;Y)$). We have equality in \eqref{conditioningreducesH} if and only if $X,Y$ are independent. 

Inequality \eqref{conditioningreducesH} has the conditional version
\begin{equation}
    h(X|Y,Z) \leq h(X|Z) 
\end{equation}
with equality if and only if $X$ and $Y$ are conditionally independent given $Z$, i.e. if and only if for a.e. $x,y,z,$ $f_{X,Y|Z}(x,y|z) = f_{X|Z}(x|z)f_{Y|Z}(y|z),$ or equivalently $f_{X|Y,Z}(x|y,z) = f_{X|Y}(x|y)$.

\end{appendices}

\bibliographystyle{IEEEtran}
\bibliography{references}

\end{document}